\documentclass[a4paper,onecolumn,11pt]{quantumarticle}
\pdfoutput=1
\usepackage[utf8]{inputenc}
\usepackage[margin=1in]{geometry}
\usepackage[english]{babel}
\usepackage[T1]{fontenc}
\usepackage{amsmath}
\usepackage{amssymb}
\usepackage{amsthm}
\usepackage{xcolor}
\usepackage[numbers,sort&compress]{natbib}
\usepackage{hyperref}

\usepackage{tikz}
\usepackage{lipsum}

\newtheorem{theorem}{Theorem}

\newtheorem{corollary}{Corollary}
\newtheorem{lemma}{Lemma}
\newtheorem{remark}{Remark}
\newtheorem{heuristic}{Heuristic}

\begin{document}

\title{Quantum Fast-Forwarding Beyond Reversibility: 
The $\alpha$-Perturbed $n$-Cycle}
\author{}
\date{}
\maketitle

\begin{center}
\begin{tabular}{ccc}
Avah Banerjee & Asim Sharma & Sooraj Soman \\
semiqlassical, Inc. & Missouri S\&T  & Missouri S\&T \\ 
Missouri S\&T &  \  & \ 
\end{tabular}
\end{center}

\begin{abstract}
Quantum fast-forwarding (QFF) is usually formulated for reversible Markov chains, where the projected quantum walk evolution is exactly governed by Chebyshev polynomials of a Hermitian discriminant matrix. We study whether this framework can be extended to nonreversible dynamics for an $\alpha$-perturbed $n$-cycle Markov chain, which preserves circulant structure while introducing controlled irreversibility.

We show that the nonreversible case has a fundamental obstruction: for $\alpha \neq 0$, the eigenvalues of $P_\alpha$ leave the interval $[-1,1]$, so $T_m(P_\alpha)$ is not uniformly bounded and cannot arise as an exact unitary compression for all times. Thus, exact Chebyshev-based QFF does not extend directly beyond reversibility.

Nevertheless, we obtain a finite-time approximation result using truncated Chebyshev and LCU techniques. The evolution $P_\alpha^t$ can be approximated with degree
$\tau=O\left(|\alpha|t+\sqrt{t\log(t/\eta)}\right),$
which recovers the reversible $O(\sqrt t)$ behavior only in the perturbative regime $|\alpha|=O(t^{-1/2})$. This identifies a nearly reversible regime where QFF survives perturbatively and quantifies how irreversibility degrades the speedup.

\end{abstract}

\section{Introduction}

Quantum walks provide a fundamental mechanism through which quantum algorithms achieve speedups over classical algorithms. They have been used in search, sampling, and graph algorithms, often providing provable advantages in hitting times and mixing behavior.

A particularly powerful development is quantum fast-forwarding (QFF), introduced by Apers \emph{et al.}, which allows simulation of a classical Markov chain for time $t$ using only $O(\sqrt{t})$ quantum-walk steps.

However, existing results rely on the reversibility of the Markov chain. In many physical and algorithmic systems, the underlying dynamics are not reversible. This motivates the central question of this work:

\begin{quote}
Does quantum fast-forwarding persist for nearly reversible Markov chains?
\end{quote}
\paragraph{Problem Statement.}
Given a nonreversible Markov chain with transition matrix $P$ and eigenvalues $\lambda_k \in\mathbb{C}$, can quantum fast-forwarding approximate $P^t$ using $o(t)$ quantum-walk steps, and how does the complexity depend on the degree of irreversibility?

\paragraph{Our Contributions.}

\begin{itemize}
\item We give the full Fourier spectral decomposition of the biased cycle and separate the odd-$n$ mixing case from the even-$n$ periodic case.
\item We show that the actual Szegedy spectral mapping for the irreversible cycle depends on the singular values $|\lambda_k|$, not the complex eigenvalues $\lambda_k$ themselves.
\item We prove a no-go theorem: for $\alpha\neq0$, no unitary walk compression can realize $T_m(P_\alpha)$ for all $m$.
\item We replace the global $t'(\alpha)$ claim by a corrected truncated-Chebyshev/LCU theorem requiring degree $O(|\alpha|t+\sqrt{t\log(t/\eta)})$.
\end{itemize}

\section{Previous and Related Works}
Quantum walks provide a systematic way to obtain quantum algorithms from classical random-walk and Markov-chain processes. A foundational construction was introduced by Szegedy, who quantized Markov-chain-based algorithms using a product of reflections on an enlarged Hilbert space~\cite{szegedy2004quantum}. In the symmetric or reversible setting, the resulting quantum walk has a spectral structure closely related to that of the original Markov chain, which allows classical hitting-time behavior to be translated into quantum-walk dynamics~\cite{szegedy2004quantum}.

This product-of-reflections viewpoint was developed further in the context of quantum-walk search. Ambainis, Gily\'en, Jeffery, and Kokainis studied the problem of finding marked vertices by quantum walks and showed that, under broad conditions, marked vertices can be found with a quadratic speedup over the corresponding classical random walk~\cite{ambainis2020quadratic}. Apers, Gily\'en, and Jeffery later gave a unified framework for quantum-walk search, connecting hitting-time, electric-network, and MNRS-type formulations~\cite{apers2019unified}. These works are relevant here because they show how the spectral structure of a Markov chain can be converted into quantum algorithmic speedups.

A major development beyond search was quantum fast-forwarding (QFF), introduced by Apers and Sarlette~\cite{apers2018quantum}. Unlike earlier quantum-walk algorithms that mainly speed up limiting behavior such as hitting or mixing, QFF accelerates the transient evolution of a reversible Markov chain. Their method uses a Szegedy-type walk, Chebyshev polynomial approximation, and linear combinations of unitaries to approximate many classical Markov-chain steps using roughly the square root of the corresponding number of quantum-walk steps~\cite{apers2018quantum}. This made QFF a natural tool for algorithms whose complexity depends directly on the runtime of an underlying random walk.

Several later works use QFF or closely related quantum walk techniques as algorithmic primitives. Apers applied QFF to graph expansion testing and seed-set methods, improving the random-walk component of graph property testers~\cite{apers2020expansion,apers2019seedsets}. Li and Shang combined generalized interpolated quantum walks with QFF to improve search algorithms for graphs with a single marked vertex~\cite{li2023improvement}. Quantum algorithms for Markov-chain mixing and sampling have also been studied in related settings, including analog quantum algorithms for mixing Markov chains and quantum sampling from target distributions~\cite{chakraborty2020analog,bencivenga2021quantum}. These works show that quantum walks and fast-forwarding ideas are useful not only for abstract Markov-chain simulation but also for search, sampling, and graph algorithms.

Recent work  has clarified both the reach and the limitations of these ideas. Sorci studied average mixing for Szegedy quantum walks associated with reversible Markov chains, giving a framework that compares the limiting quantum walk distribution with the chain being quantized~\cite{sorci2025average}. Apers and Miclo revisited the connection between quantum walks, discrete wave equations, and Chebyshev polynomials, emphasizing that the Chebyshev description is most transparent when the transition matrix has the appropriate real spectral structure~\cite{apers2024discretewave}. These developments reinforce the role of Chebyshev polynomials in reversible quantum-walk dynamics.

The most directly related post-2021 development is the work of Claudon, Piquemal, and Monmarch\'e on quantum speedups for nonreversible Markov chains~\cite{claudon2025quantum}. Their work addresses stationary sampling for nonreversible chains using modern quantum algorithmic tools and reversibilization ideas. The present paper studies a different, finite-time question: whether the Chebyshev compression underlying QFF can simulate the transient power $P_\alpha^t$ when the transition matrix has complex eigenvalues. Thus, our biased-cycle example complements that work by isolating a concrete obstruction to transplanting reversible QFF directly to a nonreversible transition matrix.

Another closely related direction is quantum matrix powering. Since Markov-chain evolution is the application of a high power of a transition matrix to an initial distribution, QFF can be viewed as a structured quantum algorithm for matrix powers. Gonz\'alez, Trivedi, and Cirac studied quantum algorithms for powering stable Hermitian matrices and connected their work to the fast-forwarding construction of Apers and Sarlette~\cite{gonzalez2021quantum}. More recently, Low and Su introduced quantum eigenvalue processing for polynomial transformations of eigenvalues of non-normal block-encoded matrices, using Faber-polynomial ideas that generalize Chebyshev approximation to complex domains~\cite{low2026quantum}. This perspective is important for the present work because reversible QFF relies heavily on real or Hermitian spectral structure, whereas nonreversible Markov chains generally lead to non-Hermitian transition matrices with complex eigenvalues.

The limitation of existing QFF methods is that they fundamentally assume reversibility. In the reversible case, the relevant discriminant matrix has spectrum in the real interval on which Chebyshev polynomials are well behaved. For a nonreversible Markov chain, the transition matrix can have complex eigenvalues, and the same polynomial mechanism may become unstable. Thus, it is not clear whether quantum fast-forwarding continues to provide a speedup once detailed balance is broken.

This work studies that question using an $\alpha$-perturbed random walk on the $n$-cycle. The model is chosen because it is a minimal nonreversible perturbation of the reversible cycle: it introduces a directional bias while preserving enough circulant structure to allow explicit spectral analysis. For nonzero perturbation, the eigenvalues move off the real interval, creating exactly the obstruction that is absent in the reversible QFF theory. Our goal is therefore to understand how far the QFF mechanism can be pushed beyond reversibility, what fails in the direct Szegedy--Chebyshev approach, and in what nearly reversible regime a fast-forwarding-type approximation can still survive.

\section{Preliminaries}

\subsection{Markov Chains}

Let $P \in \mathbb{R}^{n \times n}$ be the transition matrix of a discrete-time Markov chain on a finite state space $\{1,2,\dots,n\}$. The entry $P_{ij}$ represents the probability of transitioning from state $i$ to state $j$, and the rows of $P$ satisfy

\begin{align*}
\sum_{j=1}^{n} P_{ij} = 1, \qquad P_{ij} \ge 0.
\end{align*}

Given an initial probability distribution $p^{(0)}$, the distribution of the system after $t$ steps evolves according to

\begin{align*}
p^{(t)} = P^t p^{(0)}.
\end{align*}

\subsection{Stationary Distribution}

A probability distribution $\pi$ is called a stationary distribution of the Markov chain if it satisfies

\begin{align*}
\pi P = \pi.
\end{align*}

Intuitively, if the Markov chain is initialized in the stationary distribution, its distribution remains unchanged under the dynamics of the chain.

\subsection{Reversible Markov Chains}

A Markov chain is said to be reversible with respect to a stationary distribution $\pi$ if the detailed balance condition holds:

\begin{align*}
\pi_i P_{ij} = \pi_j P_{ji}
\end{align*}

for all states $i$ and $j$.

Reversibility implies that the Markov chain satisfies a form of time symmetry, and it plays a central role in the analysis of many quantum-walk-based algorithms, including quantum fast-forwarding.
\section{Quantum Fast-Forwarding Framework}

Quantum fast-forwarding (QFF), introduced by Apers et al.~\cite{apers2018quantum}, provides a quantum algorithmic framework for simulating the evolution of reversible Markov chains faster than classical methods. In particular, for a reversible transition matrix $P$, QFF enables the simulation of the classical evolution $P^t$ using only $O(\sqrt{t})$ applications of an associated quantum-walk operator.

The construction begins by embedding the classical transition matrix $P$ into a unitary quantum walk acting on an enlarged Hilbert space. This walk is built from reflections determined by the transition probabilities of the Markov chain and is closely related to Szegedy's quantum walk framework~\cite{szegedy2004quantum}. A central ingredient in the construction is the discriminant matrix
\[
    D = \sqrt{P\circ P^T},
\]
where $\circ$ is elementwise multiplication. When the chain is reversible, $D$ is Hermitian and has spectrum contained in $[-1,1]$.

This Hermitian structure leads to the key spectral mapping used in QFF. If $\lambda_k$ is an eigenvalue of $D$, then there exists an angle $\theta_k$ such that
\[
    \lambda_k = \cos(\theta_k),
\]
and the corresponding eigenvalues of the quantum-walk operator are
\[
    e^{\pm i\theta_k}.
\]
This relation allows the projected quantum walk evolution to be expressed in terms of Chebyshev polynomials. Since
\[
    T_m(\lambda_k) = \cos(m\theta_k),
\]
powers of the quantum walk can be used to implement polynomial approximations to the classical evolution function $\lambda^t$. This is the mechanism behind the quadratic speedup of QFF for reversible Markov chains.

The implication $\lambda_k=\cos(\theta_k)$, however, relies crucially on reversibility, or equivalently on the Hermitian nature of the discriminant matrix $D$. For nonreversible Markov chains, the transition matrix $P$ may have complex eigenvalues, and the associated discriminant-type matrix is generally non-Hermitian. Although the Szegedy product-of-reflections construction remains unitary, its eigenphases are not determined by the complex eigenvalues of $P$ through the relation $\lambda_k=\cos(\theta_k)$. Instead, the unitary walk construction is governed by singular-value information of the discriminant-type operator.

This distinction is the source of the obstruction studied in this work. In the nonreversible setting, Chebyshev polynomials applied directly to $P$ or to a non-Hermitian discriminant-type matrix need not remain uniformly bounded. Therefore, the exact Chebyshev compression that underlies reversible QFF does not extend directly to nonreversible Markov chains.

\section{$\alpha$-Perturbed $n$-Cycle Markov Chain}

We consider a random walk on the cycle graph with $n$ vertices labeled 
$\{0,1,\dots,n-1\}$, where addition of indices is taken modulo $n$. 
From each vertex the walk moves either to the next vertex or to the 
previous vertex along the cycle.

\subsection{The Reversible $n$-Cycle Markov Chain}

In the unbiased case the transition probabilities are

\begin{align*}
P_{v,v+1} = \frac{1}{2}, \qquad
P_{v,v-1} = \frac{1}{2},
\end{align*}

for every vertex $v$. All other entries of $P$ are zero.

This transition matrix describes a simple random walk on the cycle. 
The stationary distribution of this chain is the uniform distribution

\begin{align*}
\pi_v = \frac{1}{n}, \quad v=0,1,\dots,n-1.
\end{align*}

To verify reversibility we check the detailed balance condition

\begin{align*}
\pi_u P_{u,v} = \pi_v P_{v,u}.
\end{align*}

For neighboring vertices $u=v-1$ we have

\begin{align*}
\pi_{v-1} P_{v-1,v}
=
\frac{1}{n}\cdot \frac{1}{2}
=
\frac{1}{n}\cdot \frac{1}{2}
=
\pi_v P_{v,v-1}.
\end{align*}

The same holds for the other direction along the cycle. Therefore the 
detailed balance condition is satisfied and the unbiased $n$-cycle 
random walk is a reversible Markov chain.

\subsection{The $\alpha$-Perturbed $n$-Cycle}

We now introduce a small bias parameter $\alpha$ and define the 
$\alpha$-perturbed transition probabilities

\begin{align*}
P_{v,v+1} = \frac{1}{2} + \alpha, \qquad
P_{v,v-1} = \frac{1}{2} - \alpha,
\end{align*}

where $|\alpha| < \frac{1}{2}$ so that the entries remain valid 
probabilities.

This perturbation introduces a directional bias in the random walk 
around the cycle. The stationary distribution of this chain remains 
uniform:

\begin{align*}
\pi_v = \frac{1}{n}.
\end{align*}

However, the detailed balance condition no longer holds. For neighboring 
vertices we obtain

\begin{align*}
\pi_{v-1} P_{v-1,v}
=
\frac{1}{n}\left(\frac{1}{2}+\alpha\right),
\end{align*}

while

\begin{align*}
\pi_v P_{v,v-1}
=
\frac{1}{n}\left(\frac{1}{2}-\alpha\right).
\end{align*}

These quantities are equal only when $\alpha = 0$. Therefore, for 
$\alpha \neq 0$ the detailed balance condition is violated and the 
$\alpha$-perturbed $n$-cycle defines a nonreversible Markov chain.

In this work, we investigate how quantum fast-forwarding behaves when 
applied to this nearly reversible but nonreversible Markov chain.

\subsection{Eigenvalues of the Transition Matrix}

The transition matrix of the $\alpha$-perturbed $n$-cycle is circulant, 
since each row is obtained from the previous row by a cyclic shift. 
Circulant matrices are diagonalized by the discrete Fourier basis~\cite{gray2006toeplitz}.

\begin{align*}
\psi_k(v) = e^{2\pi i k v / n}, \qquad k=0,1,\dots,n-1.
\end{align*}

Applying the transition matrix to $\psi_k$ gives

\begin{align*}
(P_\alpha\psi_k)(v)
&=
\left(\frac{1}{2}+\alpha\right)\psi_k(v+1)
+
\left(\frac{1}{2}-\alpha\right)\psi_k(v-1).
\end{align*}

Using the Fourier form of the eigenvectors we obtain

\begin{align*}
(P_\alpha\psi_k)(v)
&=
\left(\frac{1}{2}+\alpha\right)e^{i\theta_k}\psi_k(v)
+
\left(\frac{1}{2}-\alpha\right)e^{-i\theta_k}\psi_k(v),
\end{align*}

where

\begin{align*}
\theta_k = \frac{2\pi k}{n}.
\end{align*}

Thus the eigenvalues of the transition matrix are

\begin{align*}
\lambda_k =
\left(\frac{1}{2}+\alpha\right)e^{i\theta_k}
+
\left(\frac{1}{2}-\alpha\right)e^{-i\theta_k}.
\end{align*}

This expression simplifies to

\begin{align*}
\lambda_k =
\cos(\theta_k) - 2 i \alpha \sin(\theta_k).
\end{align*}

When $\alpha=0$ the eigenvalues reduce to

\begin{align*}
\lambda_k = \cos(\theta_k),
\end{align*}

which are purely real as expected for reversible Markov chains. For 
$\alpha \neq 0$ the eigenvalues become complex, reflecting the 
nonreversible nature of the perturbed dynamics.
Depending on whether the Fourier vector is chosen as $e^{+i\theta_k v}$ or $e^{-i\theta_k v}$, the imaginary part above changes sign.  All results below use only $|\lambda_k|$ or conjugate-paired modes, so the sign convention does not affect the conclusions.

\begin{theorem}[Classical spectral decomposition and odd/even cycle behavior]
Let $P_\alpha$ denote the biased walk on the $n$-cycle with
\[
p=\frac12+\alpha,\qquad q=\frac12-\alpha,\qquad |\alpha|<\frac12.
\]
Let
\[
|f_k\rangle=\frac1{\sqrt n}\sum_{x=0}^{n-1}e^{i\theta_k x}|x\rangle,
\qquad
\theta_k=\frac{2\pi k}{n}.
\]
Then $P_\alpha$ is normal and Fourier diagonalizable:
\[
P_\alpha |f_k\rangle=\lambda_k |f_k\rangle,
\qquad
\lambda_k=\cos\theta_k-2i\alpha\sin\theta_k,
\]
up to the sign convention noted above.  Moreover,
\[
|\lambda_k|^2
=
\cos^2\theta_k+4\alpha^2\sin^2\theta_k
=
1-(1-4\alpha^2)\sin^2\theta_k.
\]
If $n$ is odd, then $|\lambda_k|<1$ for every $k\neq0$, and the classical chain converges to the uniform distribution.  More precisely, for a starting vertex $v$,
\[
\left\|P_\alpha^t|v\rangle-\frac1n\mathbf 1\right\|_2^2
=
\frac1n\sum_{k=1}^{n-1}|\lambda_k|^{2t}.
\]
Consequently,
\[
\left\|P_\alpha^t|v\rangle-\frac1n\mathbf 1\right\|_2
\le
\sqrt{\frac{n-1}{n}}\,
\rho_*^t,
\]
where
\[
\rho_*
=
\sqrt{1-(1-4\alpha^2)\sin^2\left(\frac{\pi}{n}\right)}.
\]
If $n$ is even, then $\lambda_{n/2}=-1$, so the chain has a period-two obstruction and does not converge to the uniform distribution.
\end{theorem}

\begin{proof}
The Fourier diagonalization follows from circulant structure.  The displayed expression for $|\lambda_k|^2$ is immediate.  If $n$ is odd, the only value of $k$ with $\sin\theta_k=0$ is $k=0$, so every nontrivial Fourier mode has modulus strictly smaller than one.  Expanding $|v\rangle$ in the Fourier basis and removing the stationary $k=0$ mode gives
\[
P_\alpha^t|v\rangle-\frac1n\mathbf 1
=\frac1{\sqrt n}\sum_{k=1}^{n-1}e^{-i\theta_kv}\lambda_k^t|f_k\rangle,
\]
which proves the exact $\ell_2$ identity.  The bound follows by taking the largest nontrivial modulus.  If $n$ is even, $k=n/2$ gives $\theta_k=\pi$ and $\lambda_k=-1$, so the corresponding Fourier component never decays.
\end{proof}
\section{Quantum Fast-Forwarding for the $\alpha$-Perturbed $n$-Cycle}

In the standard formulation of quantum fast-forwarding (QFF), the underlying Markov
chain is assumed to be reversible. Under this assumption the discriminant matrix
associated with the transition matrix is symmetric, which allows the spectral
analysis of the quantum-walk operator to be directly related to the spectrum of
the classical Markov chain.

However, many Markov chains arising in practice are not exactly reversible.
To study the robustness of QFF beyond detailed balance, we consider the notion
of nearly reversible Markov chains.

\subsection{Nearly Reversible Markov Chains}

Let $P$ be the transition matrix of a finite, irreducible Markov chain with
stationary distribution $\pi$. The time-reversed Markov chain with respect to
$\pi$ is defined as

\begin{equation}
P^{r}_{ji} = \frac{\pi_i P_{ij}}{\pi_j}.
\end{equation}

If $P = P^{r}$, the chain satisfies the detailed balance condition and is
reversible.

We say a Markov chain is $\alpha$-nearly reversible if its transition matrix can be written as
\[
P_\alpha=P_0+\alpha A,
\]

where $P_0$ is a reversible transition matrix, $\alpha$ is a scalar perturbation parameter, and $A$ is the skew-circulant nearest-neighbor perturbation satisfying
\[
A_{v,v+1}=1,
\qquad
A_{v,v-1}=-1,
\]
and all other entries are zero. Thus the irreversibility parameter is the scalar $\alpha$, while $A$ is the fixed direction-bias matrix. In this work, we study a concrete instance of such a perturbation on the $n$-cycle.

\subsection{Generalized Discriminant Matrix}
For the chain considered here, define the \emph{generalized discriminant matrix} entrywise by

\begin{align*}
D' = \sqrt{P_\alpha \circ (P_\alpha^{r})^{T}},
\end{align*}

where the square root is taken entrywise. Expanding the entries gives
\begin{align*}
(D')_{ij}
&= \sqrt{(P_\alpha)_{ij} (P_\alpha^{r})_{ji}} \\
&= \sqrt{(P_\alpha)_{ij} \frac{\pi_i}{\pi_j} (P_\alpha)_{ij}} \\
&= \sqrt{\frac{\pi_i}{\pi_j}}\, (P_\alpha)_{ij}.
\end{align*}

Let $\Pi = \mathrm{diag}(\pi)$ denote the diagonal matrix of the
stationary distribution. Using this notation we obtain the matrix form

\begin{align*}
D' = \Pi^{1/2} P_\alpha \Pi^{-1/2}.
\end{align*}

This shows that $D'$ is related to the original transition matrix
through a similarity transformation.
\subsection{Spectral Relation Between $D'$ and $P_\alpha$}

An important property of the generalized discriminant matrix is that it shares
the same spectrum as the original transition matrix. In particular,
\begin{align*}
\mathrm{Spec}(D') = \mathrm{Spec}(P_\alpha).
\end{align*}

This follows from the fact that $D'$ is obtained from $P_\alpha$ through a similarity
transformation. Indeed,
\begin{align*}
D' = \Pi^{1/2} P_\alpha \Pi^{-1/2},
\end{align*}

which implies that $D'$ and $P_\alpha$ are similar matrices and therefore have identical
eigenvalues.

This property allows us to analyze the spectral behavior of the nearly
reversible Markov chain through the generalized discriminant matrix while
retaining the eigenvalue structure of the original transition matrix $P_\alpha$.

For the biased cycle the stationary distribution is uniform, so $D'=P_\alpha$.  Thus $D'$ is generally non-Hermitian and has complex eigenvalues.  This is useful for tracking the classical spectrum, but it is not enough to inherit the reversible Chebyshev identity.  In contrast, the ordinary reversible discriminant $\sqrt{P_\alpha\circ P_\alpha^T}$ has off-diagonal weights $\sqrt{pq}$ and is symmetric, but it no longer has the same spectrum as $P_\alpha$.  This split between ``same spectrum'' and ``unitary Chebyshev walk'' is precisely the issue addressed below.

\subsection{Modified Walk Operators}

To construct a quantum walk corresponding to a nearly reversible Markov chain, we first consider a modification of the construction of Apers and Sarlette~\cite{apers2018quantum}.
Let $P_\alpha$ be the transition matrix of the $\alpha$-perturbed Markov chain and let $P_\alpha^r$ be its time reversal with respect to the stationary distribution $\pi$.

The quantum walk acts on an extended Hilbert space
\begin{align*}
\mathcal{H} = \text{span}\{ |i,j\rangle : i,j \in S \},
\end{align*}
where the first register represents the current vertex and the second
register represents the coin state.

We define two coin operators corresponding to the transition matrices
$P_\alpha$ and $P_\alpha^r$.
\begin{align*}
V |i,b\rangle = \sum_j \sqrt{(P_\alpha)_{ij}}\,|i,j\rangle,
\end{align*}

\begin{align*}
V^r |i,b\rangle = \sum_j \sqrt{(P_\alpha^r)_{ij}}\,|i,j\rangle,
\end{align*}

where $|b\rangle$ denotes a fixed reference state in the coin register.

Next, we define the swap operator
\begin{align*}
S |i,j\rangle = |j,i\rangle.
\end{align*}

Using these operators, we construct the modified quantum-walk operator

\begin{align*}
U' = V^\dagger S V^r.
\end{align*}

This operator generalizes the standard Szegedy walk. When the Markov
chain is reversible, we have $P = P^r$ and therefore $V = V^r$,
reducing the construction to the usual quantum-walk operator

\begin{align*}
U = V^\dagger S V.
\end{align*}

Finally, the fast-forwarding walk operator is obtained by combining the
walk operator with a reflection about the flat subspace. Let
$\Pi_b$ denote the projector onto the flat subspace

\begin{align*}
\mathcal{H}_b = \text{span}\{ |i\rangle \otimes |b\rangle \}.
\end{align*}

The reflection operator is defined as

\begin{align*}
R_b = 2\Pi_b - I.
\end{align*}

The resulting fast-forwarding walk operator for nearly reversible chains
is then

\begin{align*}
W' = R_b U'.
\end{align*}

When the chain is reversible, this reduces to the standard
fast-forwarding walk operator
\begin{align*}
W = R_b U.
\end{align*}
\vspace{0.5em}
\begin{lemma}
For any quantum state $|v,b\rangle$,
\begin{align*}
\Pi_b U^\prime |v,b\rangle= D^\prime |v,b\rangle
\end{align*}
\end{lemma}
\begin{proof}
This follows directly from the fact that, for any node $i$,
\[
U^\prime |i,b\rangle
=\sum_j \sqrt{(P_\alpha)_{ij}(P_\alpha^r)_{ji}}\,|j,b\rangle+|\psi^{\perp}\rangle.
\]
Using the definition of time reversal and the uniformity of the stationary distribution, $\pi_i=\pi_j$ for all $i,j$, gives
\[
U^\prime |i,b\rangle
=\sum_j (P_\alpha)_{ij}|j,b\rangle+|\psi^{\perp}\rangle.
\]
Projecting onto the flat subspace gives the claim.
\end{proof}
This identity suggests a direct quantum-walk procedure for preparing a state proportional to $(D^{\prime})^t|v\rangle$: apply the single-step walk $t$ times and measure with $\Pi_b$ and $I-\Pi_b$~\cite{apers2018quantum}. However, this direct approach still requires $t$ walk steps. To obtain a quadratic speedup in the reversible setting, one instead uses the walk operator $W$.

The flat-subspace definition above is convenient for perturbative calculations, but it hides a unitary-completion problem: $V^\dagger S V^r$ is naturally a compression between embedded subspaces, while a full quantum walk must be unitary on the enlarged Hilbert space.  To avoid relying on a noncanonical completion, we can use the Szegedy product-of-reflections formalism.  This gives a clean spectral decomposition and shows exactly which spectral data a unitary walk can encode.

\begin{theorem}[Szegedy spectral decomposition for the biased cycle]
Let $P_\alpha$ be the biased cycle and let $P_\alpha^r=P_{-\alpha}$ be its time reversal with respect to the uniform distribution.  Define isometries
\[
A|i\rangle
=
|i\rangle\sum_j\sqrt{(P_\alpha)_{ij}}\,|j\rangle,
\qquad
B|j\rangle
=
\sum_i\sqrt{(P_\alpha^r)_{ji}}\,|i,j\rangle.
\]
Then
\[
A^\dagger B=P_\alpha.
\]
Let
\[
\mathcal W_\alpha=(2BB^\dagger-I)(2AA^\dagger-I)
\]
be the Szegedy walk.  Since $P_\alpha$ is normal, write
\[
\lambda_k=\sigma_k e^{i\varphi_k},
\qquad
\sigma_k=|\lambda_k|.
\]
For every $k$ with $0<\sigma_k<1$, define
\[
|a_k\rangle=A(e^{i\varphi_k}|f_k\rangle),
\qquad
|b_k\rangle=B|f_k\rangle.
\]
Then $\langle a_k|b_k\rangle=\sigma_k$, the space
\[
\mathcal K_k=\operatorname{span}\{|a_k\rangle,|b_k\rangle\}
\]
is invariant under $\mathcal W_\alpha$, and the eigenvalues of $\mathcal W_\alpha$ on $\mathcal K_k$ are
\[
\boxed{
 e^{\pm 2i\arccos\sigma_k}
 =
 e^{\pm 2i\arccos|\lambda_k|}.
}
\]
One choice of corresponding eigenvectors is
\[
\boxed{
|\omega_k^\pm\rangle
=
\frac{|a_k\rangle-e^{\pm i\arccos\sigma_k}|b_k\rangle}
{\sqrt2\sin(\arccos\sigma_k)}.
}
\]
The exceptional modes with $\sigma_k=1$ are one-dimensional reflection-degenerate modes.  In particular, $k=0$ is the stationary mode, and when $n$ is even the mode $k=n/2$ gives the classical period-two obstruction.
\end{theorem}

\begin{proof}
The identity $A^\dagger B=P_\alpha$ follows by direct evaluation of matrix elements:
\[
\langle i|A^\dagger B|j\rangle
=
\sqrt{(P_\alpha)_{ij}(P_\alpha^r)_{ji}}
=
(P_\alpha)_{ij},
\]
where the last equality uses the uniform stationary distribution and the definition of time reversal.  Since $P_\alpha$ is normal and Fourier diagonalizable, its singular vectors may be chosen from the Fourier basis.  If $P_\alpha|f_k\rangle=\lambda_k|f_k\rangle=\sigma_ke^{i\varphi_k}|f_k\rangle$, then $u_k=e^{i\varphi_k}|f_k\rangle$ and $v_k=|f_k\rangle$ form a singular-vector pair.  Thus
\[
\langle a_k|b_k\rangle
=\langle u_k|A^\dagger B|v_k\rangle
=\langle u_k|P_\alpha|v_k\rangle
=\sigma_k.
\]
The product of two reflections in the plane spanned by two unit vectors with inner product $\sigma_k=\cos\vartheta_k$ is a rotation by angle $2\vartheta_k$, where $\vartheta_k=\arccos\sigma_k$.  This gives the claimed eigenvalues and eigenvectors.
\end{proof}

This theorem is the main structural correction: the unitary Szegedy walk sees $|\lambda_k|$, not the complex number $\lambda_k$ through a relation $\lambda_k=\cos\theta_k$.  The irreversible phase $\varphi_k$ appears in the singular vectors/eigenvectors, not in the eigenphase formula.  Hence the reversible Chebyshev identity cannot simply be transplanted to $P_\alpha$.

\vspace{0.5em}
\begin{lemma}[Finite-time perturbative Chebyshev estimate]
For any $t \ge 1$ and any vertex state $|v,b\rangle$,
\begin{align}
\left\|
\Pi_b (W^{\prime})^t |v,b\rangle - T_t(D^{\prime})|v,b\rangle
\right\|
\le
C\,|\alpha|\, t^2,
\end{align}

for some constant $C > 0$ independent of $t$ and $n$, provided $|\alpha|$ is sufficiently small.
\end{lemma}
\vspace{0.5em}
\begin{proof}
Using the reversible identity from Apers and Sarlette~\cite{apers2018quantum},
\begin{align*}
\Pi_b W^t |v,b\rangle = T_t(D)|v,b\rangle,
\end{align*}
where $W$ and $D$ are the walk operator and discriminant matrix for the $n$-cycle Markov chain with $\alpha=0$.
\begin{align*}
\Pi_b (W^{\prime})^t |v,b\rangle - T_t(D^{\prime})|v,b\rangle =
\Pi_b\big((W^{\prime})^t - W^t\big)|v,b\rangle
+
\big(T_t(D) - T_t(D^{\prime})\big)|v,b\rangle.
\end{align*}

Using the telescoping identity,
\begin{align*}
(W^{\prime})^t - W^t
=
\sum_{k=0}^{t-1} (W^{\prime})^k (W^{\prime} - W) W^{t-1-k},
\end{align*}
and the fact that $W$ and $W^{\prime}$ are unitary, we obtain
\begin{align*}
\|(W^{\prime})^t - W^t\|
\le
t\,\|W^{\prime} - W\|.
\end{align*}

From the construction of $W^{\prime}$,
\begin{align*}
\|W^{\prime} - W\| \le C_1 |\alpha|,
\end{align*}
and hence
\begin{align*}
\|(W^{\prime})^t - W^t\| \le C_1 |\alpha| t.
\end{align*}

Since projection does not increase norm,
\begin{align}
\left\|
\Pi_b\big((W^{\prime})^t - W^t\big)|v,b\rangle
\right\|
\le
C_1 |\alpha| t.
\end{align}
Now consider the Chebyshev perturbation term. By the mean-value theorem,
\begin{align*}
\|T_t(D^{\prime}) - T_t(D)\|
\le
\sup_{x \in [-1,1]} |T_t'(x)| \cdot \|D^{\prime} - D\|.
\end{align*}

From the standard Chebyshev derivative bound,
\begin{align*}
\sup_{x \in [-1,1]} |T_t'(x)| \le t^2.
\end{align*}

Also, by a Frobenius-norm estimate for the perturbation of the transition matrix,
\begin{align*}
\|D^{\prime} - D\|_F \le C_2 |\alpha|.
\end{align*}

Thus,
\begin{align}
\|T_t(D^{\prime}) - T_t(D)\|
\le
C_2 |\alpha| t^2.
\end{align}
Combining the two bounds,
\begin{align}
&\left\|
\Pi_b (W^{\prime})^t |v,b\rangle - T_t(D^{\prime})|v,b\rangle
\right\|\le
C_1 |\alpha| t + C_2 |\alpha| t^2.
\end{align}

For $t \ge 1$, the dominant term is $O(|\alpha|t^2)$, so
\begin{align}
\left\|
\Pi_b (W^{\prime})^t |v,b\rangle - T_t(D^{\prime})|v,b\rangle
\right\|
\le
C\,|\alpha| t^2.
\end{align}
This gives an $O(|\alpha|t^2)$ bound on the finite-time deviation from the reversible Chebyshev identity.
\end{proof}
The argument above should be treated as a finite-time perturbative estimate, not as a global Chebyshev identity.  Two technical points need to be tracked:
\begin{enumerate}
\item Since $D'=P_\alpha$ has complex spectrum, the mean-value bound on $[-1,1]$ is not directly sufficient for $T_t(D')-T_t(D)$ once $|\alpha|t$ is not small.  Chebyshev polynomials grow exponentially off the real interval.
\item The estimate is meaningful only in a regime where the perturbative error remains small, for example $|\alpha|t^2\ll 1$ for this crude bound, or under a sharper ellipse/tail analysis as in the corrected LCU theorem below.
\end{enumerate}
The following theorem shows that no global exact Chebyshev realization can hold for fixed nonzero $\alpha$.

\begin{theorem}[No exact Chebyshev compression for the irreversible cycle]
Let $n\ge3$ and let $P_\alpha$ be the $\alpha$-biased walk on the $n$-cycle with $0<|\alpha|<1/2$.  There is no unitary $W_\alpha$ and projection $\Pi$ such that
\[
\Pi W_\alpha^m\Pi=T_m(P_\alpha)
\]
for every integer $m\ge0$.
\end{theorem}

\begin{proof}
If $W_\alpha$ is unitary and $\Pi$ is a projection, then
\[
\|\Pi W_\alpha^m\Pi\|\le1
\]
for every $m$.  On the other hand, $P_\alpha$ is normal and diagonalized by the Fourier basis, so
\[
\|T_m(P_\alpha)\|=\max_k |T_m(\lambda_k)|.
\]
For $\alpha\neq0$ and $n\ge3$, there exists $k$ with $\sin\theta_k\neq0$, so
\[
\lambda_k=\cos\theta_k-2i\alpha\sin\theta_k
\]
is non-real and hence does not lie in $[-1,1]$.  For any $z\notin[-1,1]$,
\[
T_m(z)
=\frac12\left[
\left(z+\sqrt{z^2-1}\right)^m
+
\left(z-\sqrt{z^2-1}\right)^m
\right],
\]
using the standard closed form for Chebyshev polynomials~\cite{mason2002chebyshev}. One of the two factors has modulus strictly larger than one, so $|T_m(z)|$ is unbounded as $m\to\infty$.  Applying this to $z=\lambda_k$ gives $\|T_m(P_\alpha)\|>1$ for some $m$, contradicting $\|\Pi W_\alpha^m\Pi\|\le1$.
\end{proof}

\begin{corollary}[Single Chebyshev iterates do not approach the odd-cycle uniform distribution]
Assume $0<|\alpha|<1/2$ and $n$ is odd.  If one uses the formal Chebyshev iterate $T_m(P_\alpha)|v\rangle$ as the projected output, then
\[
\left\|
T_m(P_\alpha)|v\rangle-\frac1n\mathbf 1
\right\|_2^2
=
\frac1n\sum_{k=1}^{n-1}|T_m(\lambda_k)|^2.
\]
Moreover, there exists a constant $c_{\alpha,n}>0$ such that
\[
\left\|
T_m(P_\alpha)|v\rangle-\frac1n\mathbf 1
\right\|_2\ge c_{\alpha,n}
\qquad\text{for all }m\ge0.
\]
Thus the single-Chebyshev-time ansatz cannot get arbitrarily close to the uniform distribution on the odd irreversible cycle.
\end{corollary}

\begin{proof}
The exact identity follows by Fourier diagonalization and removal of the $k=0$ uniform mode.  For odd $n$, every $k\neq0$ has $\sin\theta_k\neq0$, so every nontrivial $\lambda_k$ is non-real.  Since all zeros of $T_m$ are real and lie in $(-1,1)$, $T_m(\lambda_k)\neq0$ for every $m$ and every $k\neq0$.  Hence the displayed norm is positive for each fixed $m$.  The preceding theorem's proof shows that the norm is unbounded along large $m$, so the infimum over the discrete set $m\ge0$ is attained on a finite set and is strictly positive.  This gives $c_{\alpha,n}>0$.
\end{proof}

\begin{heuristic}[Local matching of the Chebyshev time index]
For the $\alpha$-perturbed $n$-cycle Markov chain, a local small-angle matching calculation suggests an $\epsilon$-approximation to the $t$-step random-walk distribution using an effective Chebyshev index $t^\prime(\alpha)$, where
\begin{align*}
    &\epsilon =O(|\alpha| t^2),\\
    &t^\prime(\alpha) = t^{\frac{1}{2} - \frac{4}{3}\frac{\sqrt{t}}{\ln t}\alpha^2 + \left(8 + \frac{16}{3}\frac{\sqrt{t}}{\ln t}\right)\alpha^4}.
\end{align*}
\end{heuristic}
\begin{proof}
We determine the Chebyshev time index $t^\prime$ that best approximates the modified walk after $t$ steps in a local perturbative regime. Using the finite-time estimate above as motivation, we match the squared-modulus eigenvalue expressions $|\lambda_k^t|^2$ and $|T_{t^\prime}(\lambda_k)|^2$ up to second order in $\alpha$ to solve for $t^\prime$. The appendix gives the small-angle expansions

\begin{align*}
|\lambda_k^t|^2 &= 1 + (-t + 4\alpha^2 t)\theta_k^2 + \Bigg[\frac{t(3t-1)}{6} + 2\alpha^2 t(t-1)\frac{6t-4}{3}+4\alpha^4 t^2 (t-1)^2 
-\frac{4\alpha^2 t^2 (3t-2)}{3}\Bigg]\theta_k^4 + O(\theta_k^6)
\end{align*}

\begin{align*}
|T_{t'}(\lambda_k)|^2 &= 1 + \left(
-t'^2 + \frac{4}{3}\alpha^2 t'^2 (2t'^2 + 1)
\right)\theta_k^2 
+ \left(
\frac{t'^4}{3}
+ \frac{\alpha^2 t'^2 (-24t'^4 - 40t'^2 + 4)}{45}
\right)\theta_k^4 
+ O(\alpha^3, \theta_k^6)
\end{align*}

The biased walk on the $n$-cycle interpolates between two extremes. When
$\alpha = 0$ the walk is fully reversible and the Chebyshev construction achieves the well-known
quantum speedup $t' = \sqrt{t}$. When $\alpha = \frac{1}{2}$ the chain is a deterministic shift around the cycle,
so one needs $t' = t$ steps to traverse length $t$. These two cases motivate an ansatz in which the effective time index scales as a power of $t$:

\begin{align*}
t'(\alpha) = t^{\beta(\alpha)}, \qquad \beta(0) = \frac{1}{2}, \quad \beta\left(\frac{1}{2}\right) = 1
\end{align*}

Because reversing the sign of the perturbation $\alpha \rightarrow -\alpha$ on the $n$-cycle merely flips the direction of the walk without changing the relevant moduli, we have

\begin{align*}
t'(\alpha) = t'(-\alpha) \quad \Longrightarrow \quad \beta(\alpha) = \beta(-\alpha)
\end{align*}

Since $\beta$ is an even function of $\alpha$, we expand it in a Taylor series under the assumption $|\alpha|\ll1$:

\begin{align*}
\beta(\alpha) = \frac{1}{2} + A\alpha^2 + O(\alpha^4)
\end{align*}

\begin{align*}
\beta(\alpha) = \frac{1}{2} + A\alpha^2 + B\alpha^4
\end{align*}

Imposing $\beta\left(\frac{1}{2}\right) = 1$ immediately yields

\begin{align*}
\frac{1}{2} + \frac{1}{4}A + \frac{1}{16}B = 1 \quad \Longrightarrow \quad 4A + B = 8.
\end{align*}

To determine $A$, we match the $O(\alpha^2 \theta^2)$ terms in the small-angle expansions

\begin{align*}
|\lambda|^{2t} = 1 - (t - 4\alpha^2 t)\theta^2 + O(\theta^4)
\end{align*}

\begin{align*}
|T_{t'}(\lambda)|^2 = 1 - t'\theta^2 + \frac{8}{3}\alpha^2 t'^2 \theta^2 + O(\theta^4)
\end{align*}

From the ansatz

\begin{align*}
t' = t^{\frac{1}{2} + A\alpha^2 + B\alpha^4} = \sqrt{t}\left(1 + A\alpha^2 \ln t + O(\alpha^4)\right), \quad
t'^2 = t\left(1 + 2A\alpha^2 \ln t + O(\alpha^4)\right)
\end{align*}

one finds

\begin{align*}
|T_{t'}(\lambda)|^2 = 1 - \sqrt{t}\theta^2 + \left[-A\sqrt{t}\ln t + \frac{8}{3}t\right]\alpha^2 \theta^2 + O(\theta^4, \alpha^4)
\end{align*}

Equating the coefficient of $\alpha^2 \theta^2$ to the classical value $+4t$ gives

\begin{align*}
-A\sqrt{t}\ln t + \frac{8}{3}t = 4t \quad \Longrightarrow \quad A = -\frac{4}{3}\frac{\sqrt{t}}{\ln t}
\end{align*}

Substituting this into $4A + B = 8$ determines

\begin{align*}
B = 8 - 4A = 8 + \frac{16}{3}\frac{\sqrt{t}}{\ln t}
\end{align*}

Hence the interpolating exponent is

\begin{align*}
\beta(\alpha) = \frac{1}{2} - \frac{4}{3}\frac{\sqrt{t}}{\ln t}\alpha^2 + \left(8 + \frac{16}{3}\frac{\sqrt{t}}{\ln t}\right)\alpha^4
\end{align*}

and hence

\begin{align*}
t'(\alpha) = t^{\frac{1}{2} - \frac{4}{3}\frac{\sqrt{t}}{\ln t}\alpha^2 + \left(8 + \frac{16}{3}\frac{\sqrt{t}}{\ln t}\right)\alpha^4}.
\end{align*}

This expression reproduces $\sqrt{t}$ at $\alpha = 0$ and $t$ at the deterministic endpoint $\alpha = \frac{1}{2}$, and, by construction, matches the $\alpha^2\theta^2$ term in the small-angle expansion.
\end{proof}
The preceding heuristic should be interpreted only as a local matching calculation. It matches a small-$\theta$, small-$\alpha$ expansion and does not control the Chebyshev growth caused by complex eigenvalues. The condition $\beta(1/2)=1$ also uses the deterministic endpoint $\alpha=1/2$, which lies outside the perturbative regime $|\alpha|\ll1$. 
\begin{theorem}[Corrected truncated-Chebyshev degree for the biased cycle]
Let $P_\alpha$ be the $\alpha$-biased walk on the $n$-cycle with $|\alpha|<1/2$.  Let
\[
F_{\tau,t}(z)=\sum_{\ell=0}^{\tau}p_\ell T_\ell(z),
\]
where
\[
p_\ell=\mathbb P(|X_t|=\ell),
\]
and $X_t$ denotes a length-$t$ simple symmetric random walk on $\mathbb Z$ starting at the origin.  Define
\[
s_\alpha=\operatorname{arsinh}(2|\alpha|),
\qquad
\rho_\alpha=e^{s_\alpha}=2|\alpha|+\sqrt{1+4\alpha^2}.
\]
Then
\[
\|F_{\tau,t}(P_\alpha)-P_\alpha^t\|_2
\le
\sum_{\ell>\tau}p_\ell\rho_\alpha^\ell.
\]
In particular, for any target error $\eta>0$, it is sufficient to choose
\[
\tau
\ge
s_\alpha t+
\sqrt{s_\alpha^2t^2+2t\log(2t/\eta)}
\]
to guarantee
\[
\|F_{\tau,t}(P_\alpha)-P_\alpha^t\|_2\le \eta.
\]
Thus
\[
\boxed{
\tau
=
O\!\left(|\alpha|t+\sqrt{t\log(t/\eta)}\right).
}
\]
The reversible $O(\sqrt t)$ scaling is recovered when $\alpha=0$, and it remains genuinely quadratic only in regimes such as $|\alpha|=O(t^{-1/2})$.
\end{theorem}

\begin{proof}
The eigenvalues of $P_\alpha$ are
\[
\lambda_k=\cos\theta_k-2i\alpha\sin\theta_k.
\]
These points lie on the Bernstein ellipse with foci $\pm1$ and parameter $\rho_\alpha$, because the ellipse
\[
z=\frac12(w+w^{-1}),
\qquad |w|=\rho_\alpha,
\]
has imaginary semiaxis $(\rho_\alpha-\rho_\alpha^{-1})/2=2|\alpha|$.  For every $z$ in this ellipse,
\[
|T_\ell(z)|\le \rho_\alpha^\ell.
\]
Since $P_\alpha$ is normal,
\[
\|F_{\tau,t}(P_\alpha)-P_\alpha^t\|_2
=
\max_k\left|\sum_{\ell>\tau}p_\ell T_\ell(\lambda_k)\right|
\le
\sum_{\ell>\tau}p_\ell\rho_\alpha^\ell.
\]
Let $s_\alpha=\log\rho_\alpha$.  The distribution of $|X_t|$ obeys the standard sub-Gaussian bound
\[
p_\ell\le 2\exp\left(-\frac{\ell^2}{2t}\right).
\]
Therefore
\[
\sum_{\ell>\tau}p_\ell\rho_\alpha^\ell
\le
2\sum_{\ell>\tau}
\exp\left(s_\alpha\ell-\frac{\ell^2}{2t}\right).
\]
For $\ell\ge\tau\ge s_\alpha t$, the exponent is decreasing in $\ell$.  Since there are at most $t$ nonzero terms,
\[
\sum_{\ell>\tau}p_\ell\rho_\alpha^\ell
\le
2t\exp\left(s_\alpha\tau-\frac{\tau^2}{2t}\right).
\]
This is at most $\eta$ whenever
\[
\frac{\tau^2}{2t}-s_\alpha\tau\ge \log(2t/\eta),
\]
which is guaranteed by the displayed choice of $\tau$.
\end{proof}

Linear combinations of unitaries (LCU) are useful because they implement polynomial transformations of quantum-walk operators that approximate Markov-chain evolution. In the reversible case, the truncation analysis uses the bound $|T_\ell(x)|\le1$ for $x\in[-1,1]$. This bound fails in the nonreversible case. In the irreversible cycle, the relevant eigenvalues are complex:
\[
z_k=\cos\theta_k-2i\alpha\sin\theta_k.
\]
Therefore, the tail is not bounded by $p_{>t'}$ alone; the corrected tail includes an additional factor $\rho_\alpha^\ell$.

\begin{lemma}[LCU tail with complex Chebyshev growth]
With $\rho_\alpha=2|\alpha|+\sqrt{1+4\alpha^2}$, the LCU truncation error satisfies
\[
\left\|
P_\alpha^t-
\sum_{\ell=0}^{\tau}p_\ell T_\ell(P_\alpha)
\right\|_2
\le
\sum_{\ell>\tau}p_\ell\rho_\alpha^\ell.
\]
Consequently, a sufficient truncation degree is
\[
\tau
=
O\!\left(|\alpha|t+\sqrt{t\log(t/\epsilon)}\right).
\]
\end{lemma}

\begin{proof}
This is the same estimate as in the corrected truncated-Chebyshev theorem.  The algebraic identity
\[
z^t=\sum_{\ell=0}^t p_\ell T_\ell(z)
\]
continues to hold for complex $z$, but truncating it requires the Bernstein ellipse bound $|T_\ell(z_k)|\le\rho_\alpha^\ell$.  Applying this uniformly over all Fourier modes gives the result.
\end{proof}

The LCU procedure remains meaningful, but the query degree is no longer governed only by the Gaussian width $\sqrt t$ of the coefficients $p_\ell$.  Irreversibility moves the spectrum off $[-1,1]$ and adds a drift-like cost $|\alpha|t$.  Therefore the QFF speedup is robust only when this extra cost does not dominate $\sqrt t$.
\subsection{Discussion}

The preceding result analyzes the performance of quantum fast-forwarding under small nonreversible perturbations. Since the spectrum leaves $[-1,1]$, Chebyshev polynomials acquire exponential growth controlled by
\[
\rho_\alpha=2|\alpha|+\sqrt{1+4\alpha^2}
\quad\text{or equivalently}\quad
s_\alpha=\log\rho_\alpha=\operatorname{arsinh}(2|\alpha|).
\]
The resulting degree bound is
\[
\tau=O\!\left(|\alpha|t+\sqrt{t\log(t/\eta)}\right).
\]
Thus the reversible quadratic speedup persists in a genuinely perturbative regime such as $|\alpha|=O(t^{-1/2})$, but for fixed nonzero $\alpha$ the term $|\alpha|t$ eventually dominates and the asymptotic quadratic speedup is lost.

\section{Conclusion}

In this work, we investigated the behavior of quantum fast-forwarding
(QFF) beyond the standard assumption of reversible Markov chains. 
Focusing on the $\alpha$-perturbed $n$-cycle, we analyzed how the
Szegedy and LCU-based Chebyshev frameworks behave when the transition
matrix becomes nonreversible.

The biased cycle exposes a sharp obstruction to extending reversible QFF by simply replacing the real discriminant with a non-Hermitian matrix sharing the spectrum of $P_\alpha$.  Although the classical odd cycle still converges to uniform for $|\alpha|<1/2$, the formal Chebyshev iterates $T_m(P_\alpha)$ do not converge to uniform and cannot be realized as exact compressions of a unitary walk for all $m$.  A positive LCU approximation remains available, and its truncation degree is
\[
O\!\left(|\alpha|t+\sqrt{t\log(t/\eta)}\right).
\]
This gives a precise perturbative regime in which QFF-like behavior survives and identifies the mechanism by which fixed nonzero irreversibility destroys asymptotic quadratic speedup.

\section{Future Work}

A sharper version of the present result would replace the coarse union-bound tail by an optimal weighted moderate-deviation estimate for
\[
\sum_{\ell>\tau}p_\ell\rho_\alpha^\ell.
\]
Another direction is to characterize nonreversible chains whose spectra remain inside a thin Bernstein ellipse around $[-1,1]$, because the ellipse parameter directly controls the LCU degree.  Finally, one should investigate whether alternative block-encodings, singular-value transformations, or eigenvalue-processing methods can simulate $P^t$ without trying to realize $T_m(P)$ for a non-Hermitian $P$~\cite{low2026quantum,sunderhauf2023generalized}.

\paragraph{Acknowledgments.}
This research was supported in part by the  National Science Foundation under Award No.~CCF-2246144. Any opinions, findings, conclusions, or recommendations expressed in this material are those of the authors and do not necessarily reflect the views of the National Science Foundation.

\bibliographystyle{quantum}
\bibliography{references}
\appendix

\section{Second-Order Expansion of $\lambda_k^t$}

We expand the power $\lambda_k^t = (\cos\theta_k + \Delta_k)^t$ up to second order in $\alpha$, where

\begin{align*}
\Delta_k = -2i\alpha \sin\theta_k
\end{align*}

Using the binomial expansion around $\alpha = 0$, we have:

\begin{align*}
\lambda_k^t = \cos^t\theta_k + t\cos^{t-1}\theta_k \Delta_k + \frac{t(t-1)}{2}\cos^{t-2}\theta_k \Delta_k^2 + O(\Delta_k^3)
\end{align*}

Substituting $\Delta_k = -2i\alpha \sin\theta_k$ gives:

\begin{align*}
\lambda_k^t = \cos^t\theta_k - 2i\alpha t \cos^{t-1}\theta_k \sin\theta_k - 2\alpha^2 t(t-1)\cos^{t-2}\theta_k \sin^2\theta_k + O(\alpha^3)
\end{align*}

Thus, the second-order expansion of $\lambda_k^t$ is:

\begin{align*}
\boxed{
\lambda_k^t = \cos^t\theta_k - 2i\alpha t \cos^{t-1}\theta_k \sin\theta_k - 2\alpha^2 t(t-1)\cos^{t-2}\theta_k \sin^2\theta_k + O(\alpha^3)
}
\end{align*}

\section{Second-Order Expansion of $T_{t'}(\lambda_k)$}

Let

\begin{align*}
z(\alpha) = \cos\theta_k - 2i\alpha \sin\theta_k, \qquad \alpha \in \mathbb{R},\ |\alpha| \ll 1
\end{align*}

Throughout, terms of order $O(\alpha^3)$ and higher are neglected. Now, Chebyshev polynomials of the first kind satisfy

\begin{align*}
T_t(z) = \frac{1}{2}\left[(z + \sqrt{z^2 - 1})^t + (z - \sqrt{z^2 - 1})^t \right], \qquad t \in \mathbb{Z}_{\ge 0},\ z \in \mathbb{C}
\end{align*}

We now compute the second-order expansion of this expression for the input $z = \cos\theta_k - 2i\alpha \sin\theta_k$. We begin by expanding the argument inside the square root:

\begin{align*}
z^2 - 1
&= (\cos\theta_k - 2i\alpha \sin\theta_k)^2 - 1 \\
&= \cos^2\theta_k - 4i\alpha \cos\theta_k \sin\theta_k - 4\alpha^2 \sin^2\theta_k - 1 \\
&= -\sin^2\theta_k (1 + 4i\alpha \cot\theta_k + 4\alpha^2)
\end{align*}

For $|u| \ll 1$ the expansion $\sqrt{1+u} = 1 + \frac{1}{2}u - \frac{1}{8}u^2 + O(u^3)$ applies. To compute the square root we set $u = 4i\alpha \cot\theta_k + 4\alpha^2$, then

\begin{align*}
\sqrt{z^2 - 1}
&= i\sin\theta_k \left(1 + \frac{1}{2}u - \frac{1}{8}u^2 \right) + O(\alpha^3) \\
&= i\sin\theta_k \left(1 + 2i\alpha \cot\theta_k + 2\alpha^2 - \frac{1}{8}(16\alpha^2 \cot^2\theta_k + 32i\alpha^3 \cot\theta_k + 16\alpha^4)\right) \\
&= i\sin\theta_k \left(1 + 2i\alpha \cot\theta_k + 2\alpha^2 + 2\alpha^2 \cot^2\theta_k - 4i\alpha^3 \cot\theta_k - 2\alpha^4\right) \\
&= i\sin\theta_k - 2\alpha \cos\theta_k + 2i\alpha^2\left(\sin\theta_k + \frac{\cos^2\theta_k}{\sin\theta_k}\right) + O(\alpha^3) \\
&= i\sin\theta_k - 2\alpha \cos\theta_k + \frac{2i\alpha^2}{\sin\theta_k} + O(\alpha^3)
\end{align*}

Define

\begin{align*}
w_\pm = z \pm \sqrt{z^2 - 1}
\end{align*}

Then, using the previous expansion,

\begin{align*}
w_\pm
&= \cos\theta_k - 2i\alpha \sin\theta_k \pm \left(i\sin\theta_k - 2\alpha \cos\theta_k + \frac{2i\alpha^2}{\sin\theta_k}\right) \\
&= \cos\theta_k \pm i\sin\theta_k + 2\alpha(\pm i\sin\theta_k + \cos\theta_k) \pm \frac{2i\alpha^2}{\sin\theta_k} + O(\alpha^3) \\
&= e^{\pm i\theta_k} \left(1 + 2\alpha \pm \frac{2i\alpha^2}{\sin\theta_k e^{\pm i\theta_k}}\right) + O(\alpha^3)
\end{align*}

To raise these to the power $t'$, we use the binomial approximation, for constants $a,b$:

\begin{align*}
(1 + a\alpha + b\alpha^2)^{t'} = 1 + t'a\alpha + \left[t'b + \frac{1}{2}t'(t'-1)a^2\right]\alpha^2 + O(\alpha^3)
\end{align*}

Applying this:

\begin{align*}
w_\pm^{t'} = e^{\pm it'\theta_k}\left[1 + 2t'\alpha + \left(2t'(t'-1) \pm \frac{2it'}{\sin\theta_k e^{\pm i\theta_k}}\right)\alpha^2\right] + O(\alpha^3)
\end{align*}

Combining the two powers,

\begin{align*}
T_{t'}(z(\alpha))
&= \frac{1}{2}(w_+^{t'} + w_-^{t'}) + O(\alpha^3) \\
&= \cos(t'\theta_k) - 2i\alpha t'\sin(t'\theta_k) \\
&\quad + \alpha^2\left[2t'^2\cos(t'\theta_k) - 2t'\cot\theta_k \sin(t'\theta_k)\right] + O(\alpha^3)
\end{align*}

Therefore, the second-order expansion of $T_{t'}(\lambda_k)$ is:

\begin{align*}
\boxed{
T_{t'}(\lambda_k) = \cos(t'\theta_k) - 2i\alpha t'\sin(t'\theta_k) + 2\alpha^2 t'^2\cos(t'\theta_k) - 2\alpha^2 t'\cot\theta_k \sin(t'\theta_k) + O(\alpha^3)
}
\end{align*}
\subsection{Taylor-Series Expansions of $\lambda_k^t$ and $T_{t'}(\lambda_k)$}

We now apply small-angle approximations for the trigonometric functions to derive Taylor-series expansions of $\lambda_k^t$ and $T_{t'}(\lambda_k)$ in powers of $\theta_k$ and $\alpha$. Throughout, we assume $|\theta_k| \ll 1$ and retain terms up to $O(\alpha^2, \theta_k^4)$, unless otherwise noted.

\begin{align*}
\cos\theta_k = 1 - \frac{\theta_k^2}{2} + \frac{\theta_k^4}{24} + O(\theta_k^6), 
\qquad
\sin\theta_k = \theta_k - \frac{\theta_k^3}{6} + O(\theta_k^5), 
\qquad
\cot\theta_k = \frac{1}{\theta_k} - \frac{\theta_k}{3} - \frac{\theta_k^3}{45} + O(\theta_k^5)
\end{align*}

\subsubsection{Expansion of $\lambda_k^t$}

We begin by expanding $\lambda_k^t$ using the expression:

\begin{align*}
\lambda_k^t = \cos^t\theta_k - 2i\alpha t \cos^{t-1}\theta_k \sin\theta_k - 2\alpha^2 t(t-1)\cos^{t-2}\theta_k \sin^2\theta_k + O(\alpha^3)
\end{align*}

Substituting the Taylor expansions of $\cos\theta_k$ and $\sin\theta_k$:

\begin{align*}
\lambda_k^t
&=
\left(1 - \frac{\theta_k^2}{2} + \frac{\theta_k^4}{24} + O(\theta_k^6)\right)^t \\
&\quad
- 2i\alpha t
\left(1 - \frac{\theta_k^2}{2} + \frac{\theta_k^4}{24} + O(\theta_k^6)\right)^{t-1}
\left(\theta_k - \frac{\theta_k^3}{6} + O(\theta_k^5)\right) \\
&\quad
- 2\alpha^2 t(t-1)
\left(1 - \frac{\theta_k^2}{2} + \frac{\theta_k^4}{24} + O(\theta_k^6)\right)^{t-2}
\left(\theta_k^2 - \frac{\theta_k^4}{3} + O(\theta_k^6)\right)
+ O(\alpha^3)
\end{align*}

Next, we expand the powers using binomial approximations and collect terms:

\begin{align*}
\lambda_k^t
&=
\left(1 - \frac{t}{2}\theta_k^2 + \left(\frac{t}{24} + \frac{t(t-1)}{8}\right)\theta_k^4 + O(\theta_k^6)\right) \\
&
- 2i\alpha t
\left(1 - \frac{(t-1)}{2}\theta_k^2 + \left(\frac{(t-1)}{24} + \frac{(t-1)(t-2)}{8}\right)\theta_k^4 + O(\theta_k^6)\right)
\left(\theta_k - \frac{\theta_k^3}{6} + O(\theta_k^5)\right) \\
&
- 2\alpha^2 t(t-1)
\left(1 - \frac{(t-2)}{2}\theta_k^2 + \left(\frac{(t-2)}{24} + \frac{(t-2)(t-3)}{8}\right)\theta_k^4 + O(\theta_k^6)\right)
\left(\theta_k^2 - \frac{\theta_k^4}{3} + O(\theta_k^6)\right)
+ O(\alpha^3) \\
&= 1 - 2i\alpha t \theta_k - \left(\frac{t}{2} + 2\alpha^2 t(t-1)\right)\theta_k^2 + i\alpha t \frac{3t-2}{3}\theta_k^3 \\
&
+ \left(\frac{t}{24} + \frac{t(t-1)}{8} + \alpha^2 t(t-1)(t-2) + \frac{2}{3}\alpha^2 t(t-1)\right)\theta_k^4
+ O(\alpha^3, \theta_k^5)
\end{align*}

Finally, simplifying and collecting terms:

\begin{align*}
\lambda_k^t &= 1 - 2i\alpha t \theta_k - \left(\frac{t}{2} + 2\alpha^2 t(t-1)\right)\theta_k^2 + \frac{i\alpha t}{3}(3t-2)\theta_k^3  \\ & +\left(\frac{t}{24}(3t-2) + \alpha^2 t(t-1)\frac{3t-4}{3}\right)\theta_k^4 + O(\alpha^3, \theta_k^5)
\end{align*}

\subsubsection{Expansion of $T_{t'}(\lambda_k)$}

We now derive a Taylor-series expansion for the Chebyshev polynomial $T_{t'}(\lambda_k)$:

\begin{align*}
T_{t'}(\lambda_k) = \cos(t'\theta_k) - 2i\alpha t' \sin(t'\theta_k) + 2\alpha^2 t'^2 \cos(t'\theta_k) - 2\alpha^2 t' \frac{\cos\theta_k \sin(t'\theta_k)}{\sin\theta_k} + O(\alpha^3)
\end{align*}

Using Taylor expansions of the trigonometric functions and the cotangent:

\begin{align*}
T_{t'}(\lambda_k)
&=
\left(1 - \frac{t'^2}{2}\theta_k^2 + \frac{t'^4}{24}\theta_k^4 + O(\theta_k^6)\right)
- 2i\alpha t'
\left(t'\theta_k - \frac{t'^3}{6}\theta_k^3 + O(\theta_k^5)\right) \\
&\quad
+ 2\alpha^2 t'^2
\left(1 - \frac{t'^2}{2}\theta_k^2 + \frac{t'^4}{24}\theta_k^4 + O(\theta_k^6)\right) \\
&\quad
- 2\alpha^2 t'
\left(\frac{1}{\theta_k} - \frac{\theta_k}{3} - \frac{\theta_k^3}{45} + O(\theta_k^5)\right)
\left(t'\theta_k - \frac{t'^3}{6}\theta_k^3 + \frac{t'^5}{120}\theta_k^5 + O(\theta_k^7)\right)
\end{align*}

Combining terms and simplifying:

\begin{align*}
T_{t'}(\lambda_k)
&= 1 - 2i\alpha t'^2 \theta_k - \left(\frac{t'^2}{2} + \frac{2\alpha^2 t'^2}{3}(t'^2 - 1)\right)\theta_k^2 \\
&\quad
+ \frac{i\alpha t'^4}{3}\theta_k^3
+ \left(\frac{t'^4}{24} + \frac{\alpha^2 t'^2}{45}(3t'^4 - 5t'^2 + 2)\right)\theta_k^4
+ O(\alpha^3, \theta_k^5)
\end{align*}
\section{Modulus-Squared Expansions of $\lambda_k^t$ and $T_{t'}(\lambda_k)$}

We now derive expansions for the squared moduli $|\lambda_k^t|^2$ and $|T_{t'}(\lambda_k)|^2$, using previously obtained Taylor-series approximations in powers of $\theta_k$ and $\alpha$.

\subsection{Expansion of $|\lambda_k^t|^2$}

Starting from the earlier expansion of $\lambda_k^t$, we recall:

\begin{align*}
\lambda_k^t &= 1 - 2i\alpha t \theta_k - \left(\frac{t}{2} + 2\alpha^2 t(t-1)\right)\theta_k^2 + \frac{i\alpha t}{3}(3t-2)\theta_k^3 \\ & + \left(\frac{t}{24}(3t-2) + \alpha^2 t(t-1)\frac{3t-4}{3}\right)\theta_k^4 + O(\alpha^3, \theta_k^5)
\end{align*}

To compute the modulus squared, we separate the real and imaginary parts:

\begin{align*}
\mathrm{Re}(\lambda_k^t) = 1 + \left(-\frac{t}{2} - 2\alpha^2 t(t-1)\right)\theta_k^2 + \left(\frac{t}{24}(3t-2) + \alpha^2 t(t-1)\frac{3t-4}{3}\right)\theta_k^4,
\end{align*}

\begin{align*}
\mathrm{Im}(\lambda_k^t) = -2\alpha t \theta_k + \frac{\alpha t (3t-2)}{3}\theta_k^3
\end{align*}

Squaring each component separately, we obtain:

\begin{align*}
[\mathrm{Re}(\lambda_k^t)]^2
&= \left(1 + \left(-\frac{t}{2} - 2\alpha^2 t(t-1)\right)\theta_k^2 + \left(\frac{t}{24}(3t-2) + \alpha^2 t(t-1)\frac{3t-4}{3}\right)\theta_k^4 \right)^2 + O(\theta_k^6) \\
&= 1 + 2\left(-\frac{t}{2} - 2\alpha^2 t(t-1)\right)\theta_k^2 \\
&\quad + \left[\left(-\frac{t}{2} - 2\alpha^2 t(t-1)\right)^2 + 2\left(\frac{t}{24}(3t-2) + \alpha^2 t(t-1)\frac{3t-4}{3}\right)\right]\theta_k^4 + O(\theta_k^6) \\
&= 1 + (-t - 4\alpha^2 t(t-1))\theta_k^2 \\
&\quad + \left[\frac{t^2}{4} + 2\alpha^2 t^2(t-1) + 4\alpha^4 t^2 (t-1)^2 + \frac{t}{12}(3t-2) + 2\alpha^2 t(t-1)\frac{3t-4}{3}\right]\theta_k^4 + O(\theta_k^6) \\
&= 1 + (-t - 4\alpha^2 t(t-1))\theta_k^2 \\
&\quad + \left[\frac{t(3t-1)}{6} + 2\alpha^2 t(t-1)\frac{6t-4}{3} + 4\alpha^4 t^2 (t-1)^2 \right]\theta_k^4 + O(\theta_k^6)
\end{align*}

For the imaginary part:

\begin{align*}
[\mathrm{Im}(\lambda_k^t)]^2 = 4\alpha^2 t^2 \theta_k^2 - \frac{4\alpha^2 t^2 (3t-2)}{3}\theta_k^4 + O(\theta_k^6)
\end{align*}

Adding both contributions gives the desired modulus-squared expansion:

\begin{align*}
|\lambda_k^t|^2
&= [\mathrm{Re}(\lambda_k^t)]^2 + [\mathrm{Im}(\lambda_k^t)]^2 \\
&= 1 + (-t - 4\alpha^2 t(t-1) + 4\alpha^2 t^2)\theta_k^2 \\
&\quad + \left[\frac{t(3t-1)}{6} + 2\alpha^2 t(t-1)\frac{6t-4}{3} + 4\alpha^4 t^2 (t-1)^2 - \frac{4\alpha^2 t^2 (3t-2)}{3}\right]\theta_k^4 + O(\theta_k^6)
\end{align*}

\begin{align*}
|\lambda_k^t|^2 &= 1 + (-t + 4\alpha^2 t)\theta_k^2 + \left[\frac{t(3t-1)}{6} + 2\alpha^2 t(t-1)\frac{6t-4}{3} + 4\alpha^4 t^2 (t-1)^2 - \frac{4\alpha^2 t^2 (3t-2)}{3}\right]\theta_k^4 \\ &+ O(\theta_k^6)
\end{align*}


\subsection{Expansion of $|T_{t'}(\lambda_k)|^2$}

We now compute the squared modulus of $T_{t'}(\lambda_k)$, based on the previously derived small-angle expansion:

\begin{align*}
T_{t'}(\lambda_k) &= 1 - 2i\alpha t'^2 \theta_k + \left(-\frac{t'^2}{2} - \frac{2}{3}\alpha^2 t'^2 (t'^2-1)\right)\theta_k^2 + \frac{i\alpha t'^4}{3}\theta_k^3 \\ & + \left(\frac{t'^4}{24} + \frac{\alpha^2 t'^2}{45}(3t'^4 - 5t'^2 + 2)\right)\theta_k^4 + O(\alpha^3, \theta_k^5)
\end{align*}

Separating real and imaginary parts:

\begin{align*}
\mathrm{Re}(T_{t'}(\lambda_k)) = 1 + \left(-\frac{t'^2}{2} - \frac{2}{3}\alpha^2 t'^2 (t'^2-1)\right)\theta_k^2 + \left(\frac{t'^4}{24} + \frac{\alpha^2 t'^2}{45}(3t'^4 - 5t'^2 + 2)\right)\theta_k^4,
\end{align*}

\begin{align*}
\mathrm{Im}(T_{t'}(\lambda_k)) = -2\alpha t'^2 \theta_k + \frac{\alpha t'^4}{3}\theta_k^3
\end{align*}

We first compute the square of the real part:

\begin{align*}
[\mathrm{Re}(T_{t'}(\lambda_k))]^2
&= \left(1 + \left(-\frac{t'^2}{2} - \frac{2}{3}\alpha^2 t'^2 (t'^2-1)\right)\theta_k^2 + \left(\frac{t'^4}{24} + \frac{\alpha^2 t'^2}{45}(3t'^4 - 5t'^2 + 2)\right)\theta_k^4 \right)^2 + O(\theta_k^6) \\
&= 1 + 2\left(-\frac{t'^2}{2} - \frac{2}{3}\alpha^2 t'^2 (t'^2-1)\right)\theta_k^2 \\
&\quad + \left[\left(-\frac{t'^2}{2} - \frac{2}{3}\alpha^2 t'^2 (t'^2-1)\right)^2 + 2\left(\frac{t'^4}{24} + \frac{\alpha^2 t'^2}{45}(3t'^4 - 5t'^2 + 2)\right)\right]\theta_k^4 + O(\theta_k^6)
\end{align*}

Next, squaring the imaginary part gives:

\begin{align*}
[\mathrm{Im}(T_{t'}(\lambda_k))]^2 = 4\alpha^2 t'^4 \theta_k^2 - \frac{4}{3}\alpha^2 t'^6 \theta_k^4 + O(\theta_k^6)
\end{align*}

Adding both contributions, the modulus squared becomes:

\begin{align*}
|T_{t'}(\lambda_k)|^2
&= 1 + \left(-t'^2 + \frac{4}{3}\alpha^2 t'^2 (2t'^2+1)\right)\theta_k^2 \\
&\quad + \left(\frac{t'^4}{3} + \frac{\alpha^2 t'^2 (-24t'^4 - 40t'^2 + 4)}{45}\right)\theta_k^4 + O(\alpha^3, \theta_k^6)
\end{align*}

\begin{align*}
|T_{t'}(\lambda_k)|^2 = 1 + \left(-t'^2 + \frac{4}{3}\alpha^2 t'^2 (2t'^2 + 1)\right)\theta_k^2 + \left(\frac{t'^4}{3} + \frac{\alpha^2 t'^2 (-24t'^4 - 40t'^2 + 4)}{45}\right)\theta_k^4 + O(\alpha^3, \theta_k^6)
\end{align*}

\begin{remark}[Role of the appendix expansions]
The small-angle expansions are useful as local diagnostics, but they cannot by themselves prove a global fast-forwarding theorem. The missing global input is a bound on $T_\ell(z)$ for all eigenvalues $z=\lambda_k$. The Bernstein-ellipse estimate used in the LCU theorem supplies that control.
\end{remark}

\end{document}